\documentclass[10pt,twoside,web]{ieeecolor}
\usepackage{generic}
\usepackage{amsmath,amssymb,amsfonts,mathtools}
\usepackage[bb=boondox]{mathalfa}
\usepackage{graphicx,color,subcaption}

\newtheorem{secthm}{Theorem}[section]

\newtheorem{seclem}[secthm]{Lemma}

\newtheorem{secprop}[secthm]{Proposition}

\newtheorem{secdefn}[secthm]{Definition}

\newtheorem{secsasm}[secthm]{Standing Assumption}

\renewcommand{\theenumi}{(\roman{enumi})}

\newcommand{\0}{\mathbb{0}}
\newcommand{\I}{\mathbb{I}}

\newcommand{\bA} { {\mathbb A}}
\newcommand{\bR} { {\mathbb R}}

\def\red{\hfill $\lhd$}

\allowdisplaybreaks[4]

\def\BibTeX{{\rm B\kern-.05em{\sc i\kern-.025em b}\kern-.08em
    T\kern-.1667em\lower.7ex\hbox{E}\kern-.125emX}}
\begin{document}
\title{Contraction Analysis of Time-Delay Systems}
\author{Rintaro Watanabe and Yu Kawano
\thanks{This work was supported in part by JST FOREIS Program Grant Number JPMJFR222E and JSPS KAKENHI Grant Number JP24K00910. }
\thanks{R.~Watanabe  and Y.~Kawano are with the Graduate School of Advanced Science and Engineering, Hiroshima University, Higashi-Hiroshima 739-8527, Japan (email: m242091@hiroshima-u.ac.jp; ykawano@hiroshima-u.ac.jp).}
}
\maketitle
\begin{abstract}
In this paper, we investigate contraction analysis for nonlinear time-delay systems described by functional differential equations. We first extend the concept of Lyapunov--Krasovskii functionals within the differential framework. We then show that its existence is equivalent to that of an incremental Lyapunov--Krasovskii functional and guarantees uniform incremental {\color{black}exponential} stability. Next, we extend the concept of Lyapunov--Razumikhin functions within the differential framework, whose existence also ensures uniform incremental {\color{black}exponential} stability. As an application of our results, we formulate stabilizing feedback control design for nonlinear time-delay systems with single delays in terms of linear matrix inequalities.
\end{abstract}


\section{Introduction}
Time delays are inherent in communication networks~\cite{LSF:19}, economic systems~\cite{Gan:09}, biological processes~\cite{Mon:03}, and various physical phenomena~\cite{Han:91}. These delays can significantly influence system behavior and often induce instability. Consequently, stability analysis of time-delay systems has been an active research topic over the past few decades, particularly through the development of Lyapunov--Krasovskii and Lyapunov--Razumikhin approaches; see, e.g.,~\cite{HL:13,PK:13,EA:20} {\color{black}for systems without inputs and~\cite{Fridman:14b,Karafyllis:06,Pepe:07,IPZ:08,PKJ:17} for systems with inputs.}

In the delay-free case, incremental stability~\cite{Angeli:02}, which characterizes the stability of the distance between any pair of trajectories, has attracted considerable attention. It can be analyzed directly using incremental Lyapunov functions~\cite{Angeli:02} and differential Lyapunov functions (also called Finsler–Lyapunov functions)~\cite{FS:14} developed within the contraction framework, e.g.,~\cite{FS:14,LS:98,Bullo:22}. For incremental exponential stability, the existence of these two types of functions has been shown to be equivalent~\cite{KB:24}. Compared with equilibrium-based stability analysis, incremental stability provides stronger convergence guarantees, as the asymptotic behavior of trajectories is independent of their initial conditions. Accordingly, contraction theory offers new frameworks for analysis and design, e.g.,~\cite{MS:17,GAA:22,KSS:25,DCG:25} and enables computationally tractable methods by adequately specifying the class of systems.

In this paper, our objective is to generalize contraction analysis to nonlinear time-delay systems described by functional differential equations. Such a research direction is partly explored in, e.g.,~\cite{WS:06,LR:20,NT:18,KS:25,KH:25} as detailed by literature review below.


\smallskip

\subsubsection*{Literature Review}
Contraction analysis for time-delay systems has primarily focused on specific classes of systems or particular delay structures~\cite{WS:06,LR:20,KS:25,KH:25}. In~\cite{WS:06}, a differential 
Lyapunov--Krasovskii functional has been employed to analyze incremental stability of continuous-time delay-free systems interconnected via time-delayed communications with multiple constant delays. In~\cite{KS:25}, differential Lyapunov--Krasovskii functionals and Lyapunov--Razumikhin functions for continuous-time cooperative systems have been tailored to achieve adaptive stabilization of uncertain networks with multiple time-varying delays. In~\cite{KH:25}, a differential-type Lyapunov--Krasovskii functional has been used for contraction analysis of discrete-time stochastic monotone systems with multiple constant delays.

As an alternative approach, in~\cite{LR:20}, contraction analysis has been carried out directly using solutions for continuous-time linear time-varying systems interconnected via time-delayed nonlinear communications with arbitrary delays. A similar approach has been adopted in~\cite{NT:18}, which provides an exponential stability condition for general continuous-time time-delay systems. On the other hand, in~\cite{CPR:13}, an incremental Lyapunov--Razumikhin function has been employed to derive an incremental stability condition for general continuous-time time-delay systems. Both~\cite{NT:18} and~\cite{CPR:13} combine incremental input-to-state stability of the delay-free system with a small-gain argument. Also, an incremental Lyapunov--Krasovskii functional has been used to derive a sufficient condition for incremental input-to-state stability~\cite{PPD:10}.

{\color{black}In summary, the existing differential Lyapunov approaches have been tailored to specific system structures with constant/time-varying delays only. The existing differential and incremental Lyapunov frameworks focus on asymptotic stability properties. In contrast to the delay-free case, the connection between differential and incremental Lyapunov frameworks has not been established.}


\smallskip

\subsubsection*{Contribution}
{\color{black}In this paper, we develop differential Lyapunov frameworks for studying incremental exponential stability of general time-delay systems. Unlike the aforementioned differential Lyapunov approaches~\cite{WS:06,KS:25,KH:25}, we do not restrict our attention to specific structures such as delay-free systems interconnected via time-delayed communications or cooperative/monotone systems, nor do we assume constant/time-varying delays. For Lyapunov--Krasovskii functional analysis, we first show that a necessary and sufficient condition for exponential stability at an equilibrium point~\cite{PK:13}, under a global Lipschitz assumption, can naturally be extended to incremental exponential stability via incremental Lyapunov--Krasovskii functional analysis, again yielding a necessary and sufficient condition under a global Lipschitz assumption. 

Then, as one of the main results, we establish the equivalence between the existence of incremental and differential Lyapunov--Krasovskii functionals. In particular, inspired by the delay-free case~\cite{KB:24}, we provide an explicit procedure to construct an incremental (resp.\ a differential) Lyapunov--Krasovskii functional from a differential (resp.\ an incremental) one. The main challenge compared with the delay-free case is that the analysis must be carried out on a functional space. In particular, even variational systems are defined through Fr\'echet derivatives, requiring mathematically delicate arguments. That such extensions are nontrivial is evidenced by the substantial body of literature devoted to generalizing Lyapunov analysis from delay-free to time-delay systems. On the other hand, by combining the Lyapunov--Razumikhin framework for exponential stability analysis at an equilibrium point~\cite{EA:20} with our approach bridging incremental and differential frameworks for time-delay systems, we derive a sufficient differential Lyapunov--Razumikhin condition for incremental exponential stability.}

Finally, we illustrate the proposed framework through stabilizing controller design for nonlinear time-delay systems with a single delay. Extending existing results for linear systems~\cite{Fridman:14b}, we formulate a stabilizing controller synthesis condition within the framework of linear matrix inequalities (LMIs). {\color{black}That is, by virtue of the differential framework, we provide a numerically tractable design method for nonlinear time-delay systems, complementing incremental analysis.

The contributions of this paper are summarized as follows:
\begin{itemize}
\item we develop incremental and differential, Lyapunov--Krasovskii and 
Lyapunov--Razumikhin conditions for incremental exponential stability of 
general nonlinear time-delay systems;

\item we establish the equivalence between the existence of incremental and 
differential Lyapunov--Krasovskii functionals;

\item we formulate an LMI-based stabilizing controller synthesis condition 
for nonlinear time-delay systems.
\end{itemize}
}


\smallskip

\subsubsection*{Organization}
The rest of this paper is organized as follows. In Section~\ref{sec:Pre}, we describe the class of time-delay systems considered in this paper and review their variational systems. In Section~\ref{sec:ISA}, we present incremental and differential Lyapunov--Krasovskii functionals for time-delay systems and establish the equivalence of their existence, which guarantees uniform incremental exponential stability. We next propose a differential Lyapunov--Razumikhin functional, whose existence also ensures uniform incremental exponential stability. In Section~\ref{sec:SFD}, we apply the proposed differential Lyapunov--Krasovskii functionals to formulate stabilizing control design for nonlinear time-delay systems with single delays within the LMI framework, which is illustrated in Section~\ref{sec:Exa}. Finally, Section~\ref{sec:Con} concludes this paper.


\smallskip

\subsubsection*{Notation}
Let $(X,| \cdot |_X)$ and $(Y,| \cdot |_Y)$ be Banach spaces over the same field $K$. A linear operator $L:X \to Y$ is said to be \emph{bounded} if there exists some real number $M>0$ such that $|Lx|_Y \le M|x|_X$ for all $x \in X$. The space of all bounded linear operators from $X$ to $Y$ is denoted by $B(X,Y)$, which is a Banach space over the field $K$ with respect to the induced operator norm $|L|_{B(X,Y)} := \sup_{|x|_X=1} |Lx|_Y$.
Let $U \subset X$ be an open subset. {\color{black}The function $f: U \to Y$ is said to be \emph{continuous} at $x \in U$ if for any real number $\varepsilon > 0$, there exists some real number $\delta(\varepsilon, x) > 0$ such that
    \begin{align*}
        |x-x' |_X < \delta(\varepsilon, x), \; \forall x' \in U \quad \implies \quad | f(x)-f(x') |_Y < \varepsilon.
    \end{align*}
Also, it is said to be \emph{globally Lipschitz} on~$U$ if there exists some real number $M>0$ such that~$| f(x)-f(x') |_Y \le M |x -x'|_X$ for all~$(x,x') \in U \times U$.}
    The function $f:U \to Y$ is said to be \emph{Fr\'echet differentiable} at $x \in U$ if there exists some $D f(x) \in B(X,Y)$ such that
    \begin{align*}
        \lim_{\substack{|h |_X \to 0 \\ h \in U}} \frac{|f(x+h) - f(x) - D f(x)h|_Y}{|h |_X} = 0.
    \end{align*}
    The operator $D f(x)$ is called the \emph{Fr\'echet derivative} of $f$ at $x$ (if this exists, this is unique). 
    The function $f: U \to Y$ is said to be continuous or Fr\'echet differentiable on $V \subset U$ if $f$ is continuous or Fr\'echet differentiable at each $x \in V$, respectively. The function $f:U \to Y$ is said to be \emph{continuously Fr\'echet differentiable} at $x \in U$ if 1) it is Fr\'echet differentiable on $U$, and 2) $X \ni x \mapsto D f(x) \in B(X,Y)$ is continuous at $x$. The function $f: U \to Y$ is said to be continuously Fr\'echet differentiable on $V \subset U$ if $f$ is continuously Fr\'echet differentiable at each $x \in V$. 
    
    For $V \subset X$, let $C(V,Y)$ and $C^p(V, Y)$, $p \ge 0$ denote the set of continuous functions $f:V \to Y$ on $V$ and the set of all functions $f: V \to Y$ that have the bounded continuous Fr\'echet derivatives up through order $p$ on $V$. They are Banach spaces over the field $K$ with the supremum norm $\sup_{x\in V}|f(x)|_Y$ and with respect to the supremum norm over all Fr\'echet derivatives up through order $p$, respectively. Note that $C^0(V,Y) \subsetneq C(V,Y)$ because of the boundedness requirement for $C^0(V,Y)$.

Let $\bR$ and $\bR_+$ denote the field of real numbers and set of non-negative real numbers, respectively. A vector norm on $\bR^n$ is denoted by $| \cdot |$ regardless of its dimension $n$. Then, $(\bR^n, |\cdot |)$ is a Banach space over the field $\bR$. Let $\Omega \subset \bR^n$ be a subset. Let $C([a,b],\Omega)$ denote the set of continuous functions from $[a,b] \subset \bR$ ($a<b$) to $\Omega$ with respect to the supremum norm $\|\phi\|_C=\sup_{\theta \in [a,b]}|\phi(\theta)|$. For $\Omega = \bR^n$, $(C([a,b],\bR^n),\|\cdot \|_C)$ is a Banach space over the field $\bR$. Given $0 < r \in \bR$, if $x \in C([a - r,b],\Omega)$, then for each $t \in [a, b]$, define $x_t \in C([- r,0],\Omega)$ by $x_t (\theta) := x (t+\theta)$, $\theta \in [-r, 0]$. 

Let $U \subset \bR$ be an open subset. The function $x:U \to Y$ is said to be \emph{right-hand differentiable} at $t \in U$ if there exists a unique $\dot x(t) \in B(\bR,Y)$ (over the field $\bR$) such that
    \begin{align*}
        \dot x(t) := \lim_{\substack{h \to 0^+ \\ h \in \bR}} \frac{x(t+h) - x(t)}{h}.
    \end{align*}
The $n$-dimensional vector whose components are all zero is denoted by $\0_n$. Similarly, the $n \times m$-matrix whose components are all $0$ is denoted by $\0_{n \times m}$. The $n\times n$ identity matrix is denoted by $\I_n$.



\section{Preliminaries}\label{sec:Pre}
{\color{black}In the delay-free case, contraction analysis, e.g.,~\cite{FS:14,LS:98,Bullo:22} relies on the variational system along a solution to the original system, whose existence for all time is typically assumed. In this section, we first recall the concept of a solution and state the corresponding existence assumption for time-delay systems. We then recall the variational system associated with time-delay systems.}

\subsection{\color{black}Notion of Solution and Existence Assumption}

Let $\Omega \subset \bR^n$ be a convex subset, and let $U \subset \bR^n$ be an open subset containing $\Omega$.
Consider a time-delay system, described by the following functional differential equation:
\begin{align}\label{eq:TDS}
    \dot x(t) = f(t,x_t)
\end{align}
with $f \in C^2(\bR \times C([-r,0],U),\bR^n)$. {\color{black}Given initial data} $(t_0,\phi) \in \bR \times  C([-r,0],U)$, $\chi (t_0,\phi)(t)$ is said to be a solution to the system \eqref{eq:TDS} at time $t \in \bR$ $(t \ge t_0)$ if
\begin{enumerate}
\item $\chi_{t_0} (t_0, \phi) = \phi$;
\item there exists $0 < a(t_0, \phi) \in \bR$ such that
\begin{enumerate}
\item $x(t) := \chi (t_0,\phi)(t)$, $t \in [t_0-r, t_0+a(t_0,\phi))$ satisfies $x \in C([t_0-r, t_0+a(t_0,\phi)), U)$;
\item $x(t)$ is right-hand differentiable and satisfies \eqref{eq:TDS} for any $t \in [t_0, t_0+ a(t_0, \phi))$. 
\end{enumerate}
\end{enumerate}
We also call $\chi (t_0,\phi)(t)$ a solution on $[t_0-r, t_0+a(t_0,\phi))$ when making the interval of existence explicit.
Note that $x \in C([t_0-r, t_0+a(t_0,\phi)), U)$ implies $x_t \in C([-r, 0], U)$ for any $t \in [t_0, a(t_0, \phi))$~\cite[Lemma 2.2.1]{HL:13}. 

{\color{black}Next, we discuss the uniqueness of the solution.}
From the assumption $f \in C^2(\bR \times C([-r,0],U),\bR^n)$, for each $(t_0, \phi) \in \bR \times  C([-r,0],U)$, there exists $0 < a(t_0, \phi) \in \bR$ such that a solution $\chi (t_0,\phi)(t)$ uniquely exists on $[t_0-r, t_0+a(t_0,\phi))$ \cite[Theorem 2.4.1]{HL:13}, and $\chi (t_0,\phi)(t)$ is continuous with respect to $(t_0, \phi, t)$ on $\bR \times C([-r,0],U) \times [t_0-r,t_0+a(t_0,\phi))$ \cite[Theorem 2.2.2]{HL:13}.

Throughout this paper, we assume that for each $(t_0, \phi) \in \bR \times  C([-r,0],\Omega)$, the solution can be extended to $t \to \infty$ in the following sense, which is not explicitly mentioned hereafter.
\begin{secsasm}\label{asm:sol}
    For each $(t_0, \phi) \in \bR \times  C([-r,0],\Omega)$, 
    \begin{enumerate}
        \item the system \eqref{eq:TDS} admits a solution $\chi (t_0,\phi)(t) \in \Omega$ on $[t_0-r, \infty)$;
        \item \label{iem2} there exists $0 \le c(\phi) \in \bR$ such that $|\chi (t_0,\phi)(t)| \le c(\phi)$ for all $t, t_0 \in \bR$ $(t \ge t_0)$. \red
    \end{enumerate}
\end{secsasm}

{\color{black}Note that Standing Assumption~\ref{asm:sol} and $f \in C^2(\bR \times C([-r,0],U),\bR^n)$ implies the uniqueness of the solution on $[t_0-r, \infty)$.}


\subsection{\color{black}Variational Systems}
{\color{black}We in this subsection introduce the variational system of~\eqref{eq:TDS} along its solution, which is defined by using the Fr\'echet derivative of $f(t, \phi)$ with respect to $\phi$, denoted by $D_\phi f(t, \phi)$.} Because of $f \in C^2(\bR \times C([-r,0],\Omega),\bR^n)$, the Fr\'echet derivative uniquely exists and is continuous on $\bR \times C([-r,0],\Omega)$ as a mapping $\bR \times C([-r,0],\Omega) \ni (t, \phi) \mapsto D_\phi f(t, \phi) \in B (C([-r,0],\Omega), \bR^n)$. Under Standing Assumption~\ref{asm:sol}, $D_{\chi_t} f(t, \chi_t(t_0,\phi)) \color{black}:= D_{\upsilon} f(t, \upsilon)|_{\upsilon = \chi_t(t_0,\phi)}$ is continuous with respect to $(t_0,\phi,t)$ on $\bR \times C([-r,0],\Omega) \times [t_0-r,\infty)$. Thus, the following variational system of the system~\eqref{eq:TDS} along its trajectory $x(t) = \chi (t_0,\phi)(t)$ is well defined:
\begin{align}\label{eq:vTDS}
    \dot {\delta x}(t) = D_{\chi_t} f(t,\chi_t(t_0,\phi)) \delta x_t.
\end{align}
For $\delta \phi \in C([-r,0],\bR^n)$, its solution is defined in a similar manner as that of~\eqref{eq:TDS}. From the linearity of the right-hand side of~\eqref{eq:vTDS} with respect to $\delta x_t$ and Standing Assumption~\ref{asm:sol}, for each $(t_0, \phi) \in \bR \times  C([-r,0],\Omega)$ and every $\delta \phi \in C([-r,0],\bR^n)$, a solution $\delta x(t)$ uniquely exists on $[t_0-r,\infty)$.

It is known that the solution to the variational system~\eqref{eq:vTDS} can be constructed from the solution to the system~\eqref{eq:TDS}. This is recalled for the sake of later analysis.

\begin{seclem}\cite[Theorem 2.4.1]{HL:13}\label{lem:sol}
Consider a system~\eqref{eq:TDS} and its solution $\chi (t_0,\phi)(t)$. Then, the Fr\'echet derivative of $\chi (t_0,\phi)(t)$ with respect to $\phi$, denoted by $D_\phi \chi (t_0,\phi)(t)$, belongs to $B(C([-r,0],\bR^n),\bR^n)$ and satisfies $D_\phi \chi (t_0,\phi)(t_0) = \I_n$ for any $t \in \bR$ $(t \ge t_0)$ and $(t_0, \phi) \in \bR \times  C([-r,0],\Omega)$. Moreover, the following $\delta x(t)$:
\begin{align*}
\delta x(t) = (D_\phi \chi (t_0,\phi)(t)) \delta \phi
\end{align*}
is the solution to~\eqref{eq:vTDS} on $[t_0-r, \infty)$ for any $t \in \bR$ $(t \ge t_0)$, $(t_0, \phi) \in \bR \times  C([-r,0],\Omega)$, and $\delta \phi \in C([-r,0],\bR^n)$.
\red
\end{seclem}

Hereafter, we implicitly use the fact that $(D_\phi \chi (t_0,\phi))_t$ is equivalent to the Fr\'echet derivative of $\chi_t (t_0,\phi)$ with respect to $\phi$, denoted by $D_\phi \chi_t (t_0,\phi)$.



\section{Incremental Stability Analysis}\label{sec:ISA}

\subsection{Incremental Stability \color{black}and Incremental Lyapunov Analysis}
In this paper, we study incremental exponential stability of the time-delay system~\eqref{eq:TDS}. {\color{black}As a preliminary step toward contraction analysis, we extend in this subsection the notions of stability at an equilibrium point and the associated Lyapunov framework to the incremental setting.}

Following the delay-free case~\cite{Angeli:02,FS:14}, we introduce incremental exponential stability through an auxiliary system composed of system~\eqref{eq:TDS} and its copy, thereby characterizing the stability between pairs of trajectories:
\begin{align}\label{eq:aTDS}
    \left\{
    \begin{alignedat}{2}
        \dot x &= f(t,x_t), \quad x_{t_0} =\phi\\
        \dot x' &= f(t,x'_t), \quad x'_{t_0} =\phi',
    \end{alignedat}
    \right.
\end{align}
where $(\phi, \phi') \in C([-r,0],\Omega) \times C([-r,0],\Omega)$.

Combining {\color{black}exponential} stability of an equilibrium point for time-delay systems {\color{black}\cite[Definition 2.2]{PK:13}} and incremental stability of delay-free systems~\cite{Angeli:02,FS:14}, we introduce the following incremental {\color{black}exponential} stability for time-delay systems.

\begin{secdefn}
Consider the system~\eqref{eq:TDS} and its auxiliary system~\eqref{eq:aTDS}.
The system~\eqref{eq:TDS} is said to be {\color{black}\emph{uniformly incrementally exponentially stable} on $\Omega$ if there exist~$0 < \lambda \in \bR$ and~$1 \le k \in \bR$ such that
    \begin{align*}
    & |\chi(t_0,\phi)(t) - \chi(t_0,\phi')(t)| 
    \le k e^{-\lambda (t-t_0)} \|\phi-\phi'\|_C, 
    \; \forall t \ge t_0
    \end{align*}
    for any $t_0 \in \bR$ and $(\phi, \phi') \in C([-r,0],\Omega) \times C([-r,0],\Omega)$.
    \red
    }
\end{secdefn}

{\color{black}
In the delay-free case, there are two main approaches to studying incremental stability: 1) incremental approach, which directly employs incremental Lyapunov functions, i.e., Lyapunov functions of the auxiliary system~\eqref{eq:aTDS}; and 2) differential approach, which employs Lyapunov functions of the variational system~\eqref{eq:vTDS}. The latter often leads to numerically tractable conditions, as partly illustrated in Section~\ref{sec:SFD} below, although deriving the stability conditions is technically involved. 

Incremental Lyapunov analysis can often be developed by extending standard Lyapunov analysis at an equilibrium point. We illustrate this at the end of the subsection by extending exponential stability analysis at an equilibrium point. For time-delay systems, there are two classical Lyapunov frameworks: a) Lyapunov--Krasovskii framework, and b) Lyapunov--Razumikhin framework~\cite{PK:13,EA:20}. We present both generalizations below.

First, we begin with Lyapunov--Krasovskii analysis. As in standard analysis, we consider the upper-right Dini derivative~\cite{HL:13} of a functional~$V_I:\bR \times C([-r,0],\Omega) \times C([-r,0],\Omega) \to \bR$ along the solution to the auxiliary system~\eqref{eq:aTDS}:
\begin{align*}
    &D^+ V_I(t, \phi, \phi') :=\\
    &\limsup_{h \to 0^+} \frac{V_I(t+h, \chi_{t+h}(t,\phi), \chi_{t+h}(t,\phi')) - V_I(t, \phi, \phi')}{h}.
\end{align*}

Based on the global exponential stability analysis of an equilibrium point in~\cite[Theorem 2.5]{PK:13}, we obtain the following necessary and sufficient condition for incremental exponential stability under a global Lipschitz assumption.

\begin{secprop}\label{prop:iLK}
    A system~\eqref{eq:TDS} is uniformly incrementally exponentially stable on $\Omega$ if there exist $k_1, k_2, k_3 > 0$, $1 \le p \in \bR$, and continuous functional~$V_I:\bR \times C([-r,0],\Omega) \times C([-r,0],\Omega) \to \bR$ such that
    \begin{subequations}\label{eq2:iLK}
    \begin{align}
       & k_1 |\phi(0) - \phi'(0)|^p \le V_I(t,\phi,\phi') \le k_2 \|\phi - \phi'\|_C^p
       \label{eq2:iLK-a}\\
        &D^+ V_I(t, \phi, \phi') \le - k_3 |\phi (0) - \phi' (0)|^p
        \label{eq2:iLK-b}
    \end{align}
    \end{subequations}
    for any $t \in \bR$ ($t \ge t_0$) and $(\phi, \phi') \in C([-r,0],\Omega) \times C([-r,0],\Omega)$. The converse is also true if~$f(t, x_t)$ is globally Lipschitz with respect to~$x_t$ on~$\bR \times \Omega$. 
\end{secprop}

\begin{proof}
    In item~\ref{iem2} of Standing Assumption~\ref{asm:sol}, we assume the uniform boundedness of the solution $\chi(t_0,\phi)(t)$ with respect to $(t_0,t)$. Under this assumption, the proof is similar to that of~\cite[Theorem 2.5]{PK:13} even though the system is time-varying, since we consider the uniform stability property with respect to time.
\end{proof}

Next, we extend the Lyapunov--Razumikhin condition~\cite[Theorem 2]{EA:20} as follows. In contrast to the Lyapunov--Krasovskii framework, the converse is not clear even for standard stability at an equilibrium point.

\begin{secprop}\label{prop:iLR}
    A system~\eqref{eq:TDS} is uniformly incrementally {\color{black}exponentially} stable on $\Omega$ if there exist $k_1, k_2, k_3 > 0$, $1 \le p \in \bR$, ~$\color{black}1 < \rho \in \bR$, and continuous function~$\bar V_I: \bR \times \Omega \times \Omega \to \bR$ such that
    \begin{subequations}\label{cond:iLR}
    \begin{align}
        k_1 |x-x'|^p \le \bar V_I(t,x,x') \le k_2 |x-x'|^p
        \label{cond:iLR-a}
    \end{align}
    holds for any $t \in \bR$ and $(x, x') \in \Omega \times \Omega$, and   \begin{align}
        &\bar V_I (t+\theta, \phi(\theta), \phi'(\theta)) 
        \le \rho \bar V_I(t, \phi (0), \phi' (0)) \nonumber\\
        &\implies \quad D^+ \bar V_I(t, \phi (0), \phi'  (0)) \le  - k_3 |\phi (0) - \phi' (0)|^p \label{cond:iLR-b2} 
     \end{align}
    \end{subequations}
    holds for any $\theta \in [-r,0]$, $t \in \bR$, and $(\phi, \phi') \in C([-r,0],\Omega) \times C([-r,0],\Omega)$.
\end{secprop}
\begin{proof}
    Under item~\ref{iem2}) of Standing Assumption~\ref{asm:sol}, the proof is similar to that of~\cite[Theorem 2]{EA:20}, since we consider the uniform stability property with respect to time.
\end{proof}

Item~\ref{iem2} of Standing Assumption~\ref{asm:sol} can be removed from Propositions~\ref{prop:iLK} and~\ref{prop:iLR} if one assumes the existence of an equilibrium $x^* \in \Omega$, where by an equilibrium we mean that $\phi(t) = x^*$ for all $t \in [-r,0]$ implies $f(t,\phi)=\0_n$ for all $t \in \bR$. By fixing $\phi'(t) = x^*$ for $t \in [-r,0]$, Propositions~\ref{prop:iLK} and~\ref{prop:iLR} respectively reduce to~\cite[Theorem 2.5]{PK:13} and~\cite[Theorem 2]{EA:20}, each of which guarantees the uniform boundedness of $\chi(t_0,\phi)(t)$. 

As a further relaxation, Proposition~\ref{prop:iLK} can be extended to uniform incremental asymptotic stability by replacing the norms in~\eqref{eq2:iLK} with class~$\cal K_\infty$ functions, relying on~\cite[Theorem 2.3]{PK:13}. Such an extension of Proposition~\ref{prop:iLR} can be found in~\cite[Lemma 1]{CPR:13}.
}


\subsection{\color{black}Differential Lyapunov--Krasovskii Functionals}\label{sec:LK}

{\color{black}In this subsection, we develop a contraction framework based on the Lyapunov--Krasovskii method, which can lead to numerically tractable conditions as discussed above. The goal is to show that, for uniform incremental exponential stability, the existence of an incremental Lyapunov--Krasovskii functional and that of a differential Lyapunov--Krasovskii functional are equivalent.

To establish this equivalence, we consider an integral version of the incremental Lyapunov--Krasovskii functional from Proposition~\ref{prop:iLK}, defined as follows.
}

\begin{secdefn}
Consider the system~\eqref{eq:TDS} and its auxiliary system~\eqref{eq:aTDS}.
    A functional $V_I:\bR \times C([-r,0],\Omega) \times C([-r,0],\Omega) \to \bR$ is said to be an \emph{incremental Lyapunov--Krasovskii functional} for the system~\eqref{eq:TDS} if there exist $k_1, k_2, k_3 > 0$ and $1 \le p \in \bR$ such that
    \begin{subequations}\label{eq:iLK}
    \begin{align}
       & k_1 |\phi(0) - \phi'(0)|^p \le V_I(t,\phi,\phi') \le k_2 \|\phi - \phi'\|_C^p
       \label{eq:iLK-a}\\
        &V_I(t,\chi_t(t_0,\phi),\chi_t(t_0,\phi')) - V_I(t_0,\phi,\phi') \nonumber\\
        &\quad \le  - k_3 \int_{t_0}^t |\chi (t_0,\phi)(\tau) - \chi (t_0,\phi')(\tau)|^p d\tau
        \label{eq:iLK-b}
    \end{align}
    \end{subequations}
    for any $t, t_0 \in \bR$ ($t \ge t_0$) and $(\phi, \phi') \in C([-r,0],\Omega) \times C([-r,0],\Omega)$.
    \red
\end{secdefn}

If $V_I(t, \phi, \phi')$ is continuous, then~\eqref{eq2:iLK-b} for almost all $t$ implies~\eqref{eq:iLK-b}. Conversely, if~\eqref{eq:iLK-b} holds for all $(t_0, \phi, \phi') \in \bR \times C([-r,0],\Omega) \times C([-r,0],\Omega)$, then Standing Assumption~\ref{asm:sol} and $f \in C^2(\bR \times C([-r,0],U),\bR^n)$ guarantee~\eqref{eq2:iLK-b} everywhere.

In the delay-free case, contraction analysis provides an alternative framework for incremental stability analysis~\cite{LS:98,FS:14,Bullo:22}. The main idea is to develop an incremental stability condition via analysis of variational systems, leading to the construction of a differential Lyapunov (or called Finsler Lyapunov) function. Extending this idea to time-delay systems, we introduce a differential Lyapunov--Krasovskii functional as follows.

\begin{secdefn}
Consider the system~\eqref{eq:TDS} and its variational system~\eqref{eq:vTDS}.
    A functional $V_D: \bR \times C([-r,0],\Omega) \times C([-r,0],\bR^n) \to \bR$ is said to be a \emph{differential Lyapunov--Krasovskii functional} for the system~\eqref{eq:TDS} if there exist $k_1, k_2, k_3 > 0$ and $1 \le p \in \bR$ such that
    \begin{subequations}\label{cond:dLK}
    \begin{align}
        &k_1 |\delta \phi (0)|^p \le V_D(t,\phi,\delta \phi) \le k_2 \|\delta \phi\|_C^p
        \label{cond:dLK-a}\\
        &V_D(t,\chi_t(t_0,\phi), (D_\phi \chi (t_0,\phi))_t \delta \phi) - V_D(t_0,\phi,\delta \phi) \nonumber\\
        &\quad \le  - k_3 \int_{t_0}^t |(D_\phi \chi (t_0,\phi)(\tau)) \delta \phi|^p d\tau
        \label{cond:dLK-b}
    \end{align}
    \end{subequations}
    for any $t, t_0 \in \bR$ ($t \ge t_0$) and $(\phi,\delta \phi) \in C([-r,0],\Omega) \times C([-r,0],\bR^n)$.
    \red
\end{secdefn}

A remark analogous to the incremental case applies to the integral form~\eqref{cond:dLK-b} and its corresponding differential form.

For exponential stability of delay-free systems, it is known that an incremental Lyapunov function can be constructed from a differential Lyapunov function, and vice versa~\cite{KB:24}. As one of the main results of this paper, we generalize this to time-delay systems as follows.

\begin{secthm}\label{thm:iLK}
    Let~$1 \le p \in \bR$.
    A system~\eqref{eq:TDS} admits an incremental Lyapunov--Krasovskii functional {\color{black}on $\bR \times C([-r,0],\Omega) \times C([-r,0],\Omega)$} if it admits a continuous differential Lyapunov--Krasovskii functional {\color{black}on $\bR \times C([-r,0],\Omega) \times C([-r,0],\bR^n)$}. Conversely, a system~\eqref{eq:TDS} admits a differential Lyapunov--Krasovskii functional {\color{black}on $\bR \times C([-r,0],\Omega) \times C([-r,0],\bR^n)$} if it admits an incremental Lyapunov--Krasovskii functional  {\color{black}on $\bR \times C([-r,0],\Omega) \times C([-r,0],\Omega)$}.
\end{secthm}

\begin{proof}
The proof is in Appendix~\ref{app:iLK}.
\end{proof}

{\color{black}In contrast to Proposition~\ref{prop:iLK}, extending Theorem~\ref{thm:iLK} to class~$\mathcal{K}_\infty$ functions is not  straightforward. As shown in Appendix~\ref{app:iLK}, contraction analysis connects the variational state~$\delta \phi$ with a pair of states~$(\phi, \phi')$ via path integration. When class~$\mathcal{K}_\infty$ functions are used, it is not immediate whether the path integral~$\int_0^1 \alpha(\|d\gamma(s)/ds\|_C) ds$ can be related to a pair of states~$(\phi, \phi')$. Addressing this issue to enable asymptotic stability analysis is included in future work.}


\subsection{Lyapunov--Razumikhin Approach}

In this subsection, we propose a differential-type Lyapunov--Razumikhin function and derive an incremental {\color{black}exponential} stability condition based on it. Unlike the Lyapunov--Krasovskii framework, the connection between the incremental and differential types is not yet established.

Recalling Lemma~\ref{lem:sol}, the upper-right Dini derivative of $\bar V_D(t, \phi (0), \delta \phi (0))$ along the solutions to the system~\eqref{eq:TDS} and its variational system~\eqref{eq:vTDS} is denoted by
\begin{align*}
    &D^+ \bar V_D(t, \phi (0), \delta \phi (0)) \\
    &:=\limsup_{h \to 0^+} \frac{1}{h} \bigl(\bar V_D(t+h, \chi (t,\phi)(t+h), (D_\phi \chi (t,\phi)(t+h)) \delta \phi)\\
    &\hspace{20mm}  - \bar V_D(t, \phi(0), \delta \phi (0)) \bigr).
\end{align*}
We introduce the following differential Lyapunov--Razumikhin function.

\begin{secdefn}
    Consider the system~\eqref{eq:TDS} and its variational system~\eqref{eq:vTDS}.
    A function $\bar V_D: \bR \times \Omega \times \bR^n \to \bR$ is said to be a \emph{differential Lyapunov--Razumikhin function} for the system~\eqref{eq:TDS} if there exist $k_1, k_2, k_3 > 0$, $1 \le p \in \bR$, and~$\color{black}1 < \rho \in \bR$ such that
    \begin{subequations}\label{cond:dLR}
    \begin{align}
        k_1 |\delta x|^p \le \bar V_D(t,x,\delta x) \le k_2 |\delta x|^p
        \label{cond:dLR-a}
    \end{align}
    holds for any $t \in \bR$ and $(x, \delta x) \in \Omega \times \bR^n$, and   \begin{align}
        &\bar V_D (t+\theta, \phi(\theta), \delta \phi(\theta)) 
        \le \rho \bar V_D(t, \phi (0), \delta \phi (0)) \nonumber\\
        &\implies \quad D^+ \bar V_D(t, \phi (0), \delta \phi (0)) \le  - k_3 |\delta \phi (0) |^p \label{cond:dLR-b2} 
     \end{align}
    \end{subequations}
    holds for any $\theta \in [-r,0]$, $t \in \bR$, and $(\phi, \delta \phi) \in C([-r,0],\Omega) \times C([-r,0],\bR^n)$.
    \red
\end{secdefn}

If a system~\eqref{eq:TDS} admits a differential Lyapunov--Razumikhin function, we can show incremental exponential stability as stated below.

\begin{secthm}\label{thm:dLR}
    A system~\eqref{eq:TDS} is uniformly incrementally {\color{black}exponentially} stable on $\Omega$ if, for some $1 \le p \in \bR$, it admits a continuous differential Lyapunov--Razumikhin function {\color{black}on $\bR \times C([-r,0],\Omega) \times C([-r,0],\bR^n)$}. 
\end{secthm}
\begin{proof}
The proof is in Appendix~\ref{app:dLR}.
\end{proof}

{\color{black}By reasoning similar to that for Theorem~\ref{thm:iLK}, it is at present unclear whether Theorem~\ref{thm:dLR} can be extended to class~$\mathcal{K}_\infty$ functions. Investigating this extension, as well as converse analysis, are left for future work.}



\section{Stabilizing Feedback Design}\label{sec:SFD}

In this section, we consider the stabilization of nonlinear time-delay systems with single delays based on the proposed results for Lyapunov--Krasovskii functionals. In particular, we extend the linear matrix inequality (LMI) framework in the linear case to the nonlinear case by virtue of contraction analysis.

For stabilizing control design, it is common to use Lyapunov--Krasovskii functionals depending on the state derivatives also. The results in Section~\ref{sec:LK} can readily be generalized to such cases. Let $| \cdot |$ be the Euclidean norm, and let $W([-r,0],\Omega)$ denote the Banach space of absolutely continuous functions $\phi:[-r,0] \to \Omega$ with $\dot \phi \in L_2[-r,0]$ (the space of square integrable functions) with respect to the following norm
\begin{align*}
    \|\phi\|_W := \max_{\theta \in [-r,0]} |\phi(\theta)| + \left(\int_{-r}^0 |\dot \phi(\theta)|^2 d\theta\right)^{\frac{1}{2}}.
\end{align*}
We then reformulate the inequalities in Theorem~\ref{thm:iLK} for a functional $V:\bR \times W([-r,0],\Omega) \times L_2[-r,0] \times W([-r,0],\bR^n) \times L_2[-r,0] \to \bR$ as follows
\begin{subequations}\label{eq:LK2}
\begin{align}
    &k_1 |\delta \phi (0)|^p \le V_D(t,\phi,\dot \phi,\delta \phi,\dot{\delta \phi}) \le k_2 \|\delta \phi\|_W^p\\
    &D^+ V_D(t,\phi,\dot \phi,\delta \phi,\dot{\delta \phi}) \le - k_3 |\delta \phi (0)|^p.
\end{align}
\end{subequations}
Namely, a system~\eqref{eq:TDS} is uniformly incrementally exponentially stable on $\Omega$ if it admits a continuous functional $V_D(t,\phi,\dot \phi,\delta \phi,\dot{\delta \phi})$ satisfying~\eqref{eq:LK2}.

Based on \eqref{eq:LK2}, we design a stabilizing feedback for the following nonlinear system with a constant single-delay
\begin{align}\label{eq:sys_cd}
    \dot x(t) = f(x(t), x(t-r)) + Bu(t), \quad x_{t_0} = \phi,
\end{align}
where $\phi \in W([-r,0],\Omega)$, $r > 0$, $u(t) \in \bR^m$, and $f:\Omega \times \Omega \to \bR^n$ is of class $C^2$ and bounded. To stabilize the system~\eqref{eq:sys_cd}, we implement the following state-feedback controller
\begin{align}\label{eq:FBc}
    u(t) = F_1 x(t) + F_2 x(t-r)
\end{align}
where $F_1, F_2 \in \bR^{m \times n}$ are controller gains. 

To formulate stabilizing control design in the LMI framework, suppose that there exist finite sets of matrices $\bA_0 := \{A_0^{(1)},\dots,A_0^{(k)}\}$ and $\bA_1 := \{A_1^{(1)},\dots,A_1^{(l)}\}$ such that
\begin{align}\label{eq:conv}
\left\{
    \begin{alignedat}{2}
        \frac{\partial f(x, z)}{\partial x} \in \text{Conv}(\bA_0)\\
        \frac{\partial f(x, z)}{\partial z} \in \text{Conv}(\bA_1)
    \end{alignedat}\right.,
    \quad \forall (x, z) \in \Omega \times \Omega.
\end{align}

Now, we have the following proposition to design a stabilizing state-feedback controller for nonlinear time-delay systems.
\begin{secprop}\label{prop:BMI}
Consider a system~\eqref{eq:sys_cd} admitting a convex relaxation~\eqref{eq:conv}. Given $r \ge 0$, suppose that there exist $P, R, S \succ \0_{n \times n}$, $P_2, P_3 \in \bR^{n \times n}$, and $F_1, F_2 \in \bR^{m \times n}$ satisfying
    \begin{align}\label{eq:BMI}
        &\begin{bmatrix}
            \Phi_{11}^{(i)} & \Phi_{12}^{(i)} & \Phi_{13}^{(j)}\\
            * & -P_3 - P_3^\top + r^2R & \Phi_{23}^{(j)}\\
            * & * & -S-R
        \end{bmatrix}
        \prec \0_{3n \times 3n}
    \end{align}
    for all $i = 1, \dots, k$ and $j = 1, \dots, l$, where
    \begin{align*}
        \Phi_{11}^{(i)} &= (A_0^{(i)} + BF_1)^\top P_2 + P_2^\top (A_0^{(i)} + BF_1) + S - R\\
        \Phi_{12}^{(i)} &= P - P_2^\top + (A_0^{(i)} + BF_1)^\top P_3\\
        \Phi_{13}^{(j)} &= P_2^\top (A_1^{(j)} + BF_2) + R\\
        \Phi_{23}^{(j)} &= P_3^\top (A_1^{(j)} + BF_2)
    \end{align*}
    and $*$ denotes the symmetric term. Then, the controller~\eqref{eq:FBc} renders the closed-loop system uniformly incrementally {\color{black}exponentially} stable if $\Omega$ satisfies Standing Assumption~\ref{asm:sol} for the closed-loop system.
\end{secprop}
\begin{proof}
    Consider the following candidate of differential Lyapunov--Krasovskii functional
    \begin{align}\label{pf1:BMI}
        V_D(t,x_t,\dot x_t,\delta x_t,\dot{\delta x_t})
        &:= \delta x^\top(t) P \delta x(t) + \int_{t-r}^t \delta x^\top(s) S \delta x(s)ds \nonumber\\
        &\qquad  + r\int_{-r}^0 \int_{t+\theta}^t \dot{\delta x}^\top(s) R \dot{\delta x}(s) ds d\theta.
    \end{align}
    Then the rest is similar to the linear case~\cite[Sec 3.6.2]{Fridman:14b}.
\end{proof}

The inequality~\eqref{eq:BMI} is not an LMI because it contains nonlinear terms such as $P_2^\top BF_1$. However, this can be made an LMI by specifying $P_3 = \varepsilon P_2$ with a tunable scalar parameter $\varepsilon$.

\begin{secprop}\label{prop:LMI}
Consider a system~\eqref{eq:sys_cd} admitting a convex relaxation~\eqref{eq:conv}. Given $r \ge 0$ and $\varepsilon > 0$, suppose that there exist $\bar P, \bar R, \bar S \succ \0_{n \times n}$, $\bar P_2 \in \bR^{n \times n}$ and $X_1, X_2 \in \bR^{m \times n}$ satisfying the following LMI
    \begin{align}\label{eq:LMI}
        &\begin{bmatrix}
            \Psi_{11}^{(i)} & \Psi_{12}^{(i)} & \Psi_{13}^{(j)}\\
            * & -\varepsilon(\bar P_2^\top + \bar P_2) + r^2\bar R & \Psi_{23}^{(j)}\\
            * & * & -\bar S - \bar R
        \end{bmatrix}
        \prec \0_{3n \times 3n}
    \end{align}
   for all $i = 1, \dots, k$ and $j = 1, \dots, l$, where
    \begin{align*}
        \Psi_{11}^{(i)} &= A_0^{(i)}\bar P_2 + BX_1 + (A_0^{(i)}\bar P_2 + BX_1)^\top + S - R\\
        \Psi_{12}^{(i)} &= \bar P - \bar P_2 +\varepsilon(A_0^{(i)}\bar P_2 + BX_1)^\top\\
        \Psi_{13}^{(j)} &= A_1^{(j)}\bar P_2 + BX_2 + \bar R\\
        \Psi_{23}^{(j)} &= \varepsilon(A_1^{(j)}\bar P_2 + BX_2).
    \end{align*}
    Then, the controller~\eqref{eq:FBc} with $F_i:=X_i \bar P_2^{-1}$, $i=1,2$ renders the closed-loop system uniformly incrementally {\color{black}exponentially} stable if $\Omega$ satisfies Standing Assumption~\ref{asm:sol} for the closed-loop system.
\end{secprop}
\begin{proof}
In~\eqref{eq:BMI}, let $P_3 = \varepsilon P_2$, and define $\bar P_2 = P_2^{-1}$, $X_1 = F_1 \bar P_2$, $X_2 = F_2 \bar P_2$, $\bar P = \bar P_2^\top P \bar P_2$, $\bar R = \bar P_2^\top R \bar P_2$, and $\bar S = \bar P_2^\top S \bar P_2$. Then, multiplying~\eqref{eq:BMI} by $\text{diag}(\bar P_2, \bar P_2, \bar P_2)$ and its transpose, from right and left, respectively, leads to~\eqref{eq:LMI}.
\end{proof}

    If \eqref{eq:BMI} or \eqref{eq:LMI} is feasible for a certain delay value $r_{\max} > 0$, then this is feasible for all $0 < r \le r_{\max}$. Namely, a designed controller stabilizes a system for any delay $r$ in the interval $(0, r_{\max}]$, which provides robustness against the delay length.


\section{Example}\label{sec:Exa}
{\color{black}
Consider the nonlinear system with a constant single-delay
\begin{align}\label{eq:exa}
    \left\{
    \begin{alignedat}{2}
        \dot x_1 &= -x_1(t) + \tanh x_2(t-5) + a \sin 2t\\
        \dot x_2 &= -\sin x_1(t) - x_2(t-5) + u(t),
    \end{alignedat}\right.
\end{align}
where $a \in \bR$. 
Our objective is to design the following stabilizing controller of the form
\begin{align}\label{eq2:exa}
    u(t) = F_1 x(t) + F_2 x(t-5)
\end{align}
based on Proposition~\ref{prop:LMI}.

The variational system of the open-loop system~\eqref{eq:exa} is given by
\begin{align*}
    \begin{bmatrix}
        \dot{\delta x}_1\\
        \dot{\delta x}_2
    \end{bmatrix}
    &=
    \begin{bmatrix}
        -1 & 0\\
        -\cos x_1(t) & 0
    \end{bmatrix}
    \begin{bmatrix}
        \delta x_1(t)\\
        \delta x_2(t)
    \end{bmatrix}\\
    &+
    \begin{bmatrix}
        0 & 1 - \tanh^2 x_2(t-5)\\
        0 & -1
    \end{bmatrix}
    \begin{bmatrix}
        \delta x_1(t-5)\\
        \delta x_2(t-5)
    \end{bmatrix}
    +
    \begin{bmatrix}
        0\\
        1
    \end{bmatrix}
    \delta u(t)
\end{align*}
Since $-1 \le \cos x_1 \le 1$ and $0 \le \tanh^2 x_2 \le 1$, \eqref{eq:conv} holds for
\begin{align*}
    A_0^{(1)} &:=
    \begin{bmatrix}
        -1 & 0\\
        1 & 0
    \end{bmatrix}, \;
    A_0^{(2)} :=
    \begin{bmatrix}
        -1 & 0\\
        -1 & 0
    \end{bmatrix}\\
    A_1^{(1)} &:=
    \begin{bmatrix}
        0 & 1\\
        0 & -1
    \end{bmatrix}, \;
    A_1^{(2)} :=
    \begin{bmatrix}
        0 & 0\\
        0 & -1
    \end{bmatrix}.
\end{align*}
Solving LMI~\eqref{eq:LMI} with $\varepsilon = 3.07 \times 10^{-3}$ yields the following feedback gains
\begin{align*}
    F_1 =
    \begin{bmatrix}
        0.0514 & -1.13
    \end{bmatrix}, \;
    F_2 =
    \begin{bmatrix}
        -0.00641 & 0.959
    \end{bmatrix}.
\end{align*}
Figure~\ref{fig:ex} shows closed-loop trajectories starting from different initial data. Regardless the value of~$a$, all trajectories converge to one another, illustrating incremental exponential stability of the closed-loop system. Note that when~$a=1$, the closed-loop system does not admit an equilibrium.}

\begin{figure}
    \centering
    \includegraphics[width=1.0\linewidth]{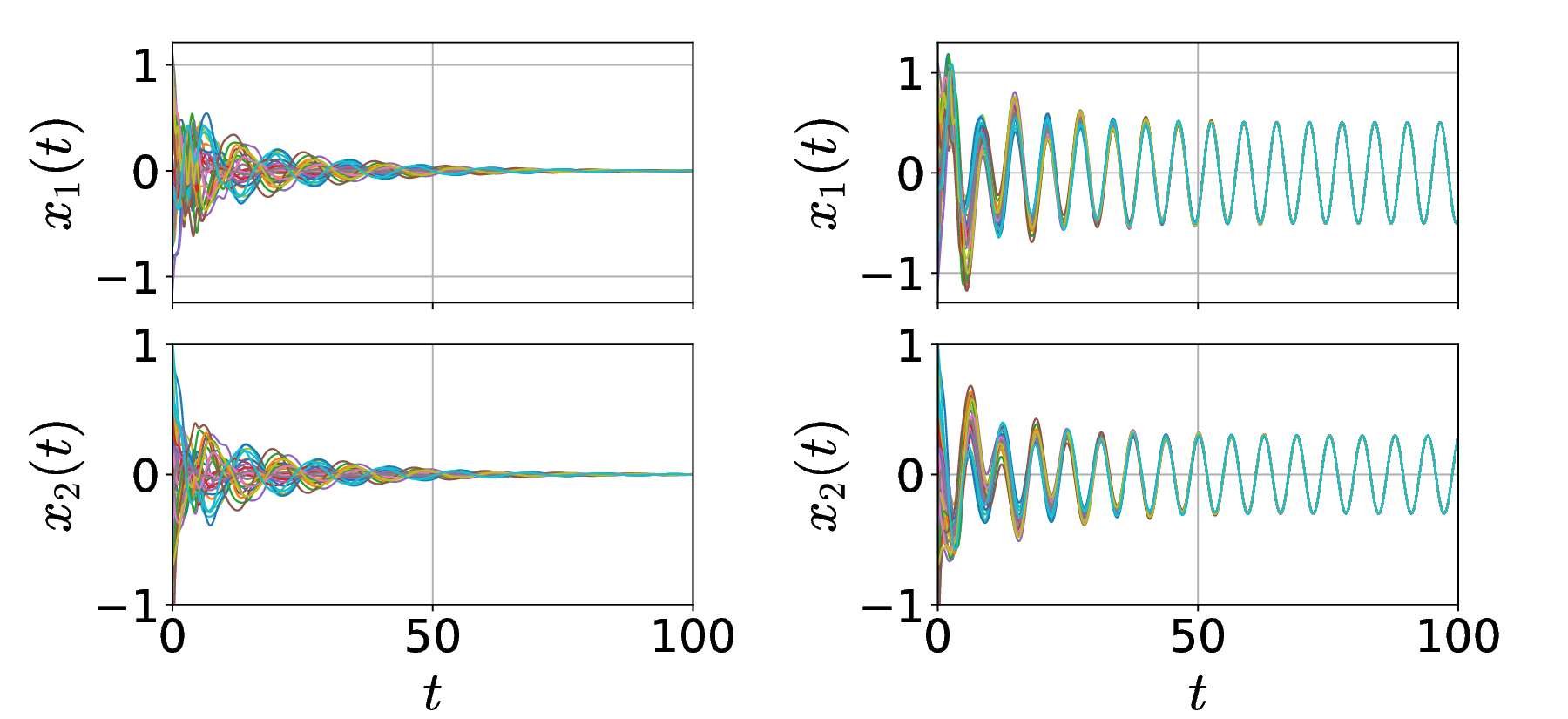}
    \caption{Closed-loop trajectories for $30$ randomly generated initial data: (left) $a=0$; (right) $a=1$.}
    \label{fig:ex}
\end{figure}




\section{Conclusion}\label{sec:Con}

In this paper, we have studied contraction analysis for nonlinear time-delay systems. We have first introduced the novel concepts of incremental and differential Lyapunov--Krasovskii functionals for time-delay systems. We have then established the equivalence between their existence, which guarantees uniform incremental exponential stability. We have next shown that the existence of a novel differential Lyapunov--Razumikhin also ensures this incremental stability. Finally, we have applied our results on Lyapunov--Krasovskii functionals to stabilizing control design, which is formulated in the LMI framework, by virtue of contraction theory. Future work includes studying the relationship between incremental and differential Lyapunov--Razumikhin functions. {\color{black}Another future direction is to extend our contraction framework to incremental input-to-state stability analysis, building on~\cite{Karafyllis:06,Pepe:07,IPZ:08,PKJ:17} for (non-incremental) stability analysis of time-delay control systems.}



\appendix

\renewcommand{\thesecthm}{\thesubsection.\arabic{secthm}}  

\subsection{Proof of Theorem~\ref{thm:iLK}}\label{app:iLK}
{\color{black}
Before showing the proof, we provide the following lemma about the path integral.
\begin{seclem}\label{lem:path}
    Let $\Gamma(\phi,\phi')$ be the collection of the class $C^1$ paths $\gamma: [0,1] \to C([-r,0],\Omega)$ connecting $\gamma(0)=\phi$ and $\gamma(1)=\phi'$, where we take the right and left derivatives at $0$ and $1$, respectively. Then, it follows that
    \begin{align}\label{eq:dist}
        \inf_{\gamma \in \Gamma(\phi,\phi')} \int_0^1 \left| \frac{d\gamma(s)}{ds}(0) \right|^p ds = |\phi(0) - \phi'(0)|^p
    \end{align}
    for any~$1 \le p \in \bR$.
\end{seclem}
\begin{proof}
        For any $\gamma \in \Gamma (\phi,\phi')$, the fundamental theorem of calculus gives
\begin{align*}
    \phi(0) - \phi'(0) = (\gamma(1))(0) - (\gamma(0))(0) = \int_0^1 \frac{d\gamma(s)}{ds} (0) ds,
\end{align*}
and consequently,
\begin{align*}
    |\phi(0) - \phi'(0)| = \left| \int_0^1 \frac{d\gamma(s)}{ds} (0) ds \right| \le \int_0^1 \left| \frac{d\gamma(s)}{ds} (0) \right| ds.
\end{align*}
From H\"older's inequality [28, Theorem~8.6], for any $p,q \in [1,\infty]$ such that $1/p + 1/q = 1$, we obtain
\begin{align*}
    \int_0^1 \left| \frac{d\gamma(s)}{ds} (0) \right| ds
    &\le \left( \int_0^1 \left| \frac{d\gamma(s)}{ds} (0) \right|^p ds \right)^{1/p} \left( \int_0^1 |1|^q ds \right)^{1/q} \nonumber\\
    &= \left( \int_0^1 \left| \frac{d\gamma(s)}{ds} (0) \right|^p ds \right)^{1/p},
\end{align*}
where for $p=\infty$, $(\int_0^1 |g(s)|^p ds)^{1/p}$ means $\sup_{s \in [0,1]}|g(s)|$.
Thus, we have
\begin{align}\label{pf1:R3}
    |\phi(0) - \phi'(0)|^p \le \int_0^1 \left| \frac{d\gamma(s)}{ds} (0) \right|^p ds.
\end{align}
Since this inequality holds for every $\gamma \in \Gamma(\phi,\phi')$, we obtain
\begin{align}\label{pf2:R3}
    |\phi(0) - \phi'(0)|^p \le \inf_{\gamma \in \Gamma(\phi,\phi')} \int_0^1 \left| \frac{d\gamma(s)}{ds} (0) \right|^p ds.
\end{align}

On the other hand, taking the line segment $\gamma(s) = s\phi + (1-s)\phi'$ yields
\begin{align*}
    \frac{d\gamma(s)}{ds}(0) = \phi(0) - \phi'(0),
\end{align*}
and thus
\begin{align*}
    \int_0^1 \left| \frac{d\gamma(s)}{ds} (0) \right|^p ds = \int_0^1 |\phi(0) - \phi'(0)|^p ds = |\phi(0) - \phi'(0)|^p.
\end{align*}
This implies
\begin{align}\label{pf3:R3}
    \inf_{\gamma \in \Gamma(\phi,\phi')} \int_0^1 \left| \frac{d\gamma(s)}{ds} (0) \right|^p ds \le |\phi(0) - \phi'(0)|^p.
\end{align}
Combining~\eqref{pf2:R3} and~\eqref{pf3:R3}, we have~\eqref{eq:dist}.
\end{proof}}

Now, we provide the proof of Theorem~\ref{thm:iLK}. 

(Differential $\Rightarrow$ Incremental) 
Let $\Gamma(\phi,\phi')$ be the collection of the paths $\gamma: [0,1] \to C([-r,0],\Omega)$ connecting $\gamma(0)=\phi$ and $\gamma(1)=\phi'$, where we take the right and left derivatives at $0$ and $1$, respectively. Our goal is to show that
    \begin{align}\label{eq:dLK2iLK}
        V_I(t,\phi,\phi') := \inf_{\gamma \in \Gamma(\phi,\phi')} \int_0^1 V_D \left( t, \gamma(s), \frac{d\gamma(s)}{ds} \right) ds
    \end{align}
    is an incremental Lyapunov--Krasovskii functional for the system~\eqref{eq:TDS}. We note that the integral in the right-hand side is well defined for any $t \in \bR$, $\gamma \in \Gamma(\phi,\phi')$, and $(\phi, \phi') \in C([-r,0],\Omega) \times C([-r,0],\Omega)$, since for each $t \in \bR$ and $(\phi, \phi') \in C([-r,0],\Omega) \times C([-r,0],\Omega)$, $V_D(t,\gamma(s),d\gamma(s)/ds)$ is a continuous function of $s \in [0,1]$.

    First, we show that \eqref{eq:dLK2iLK} satisfies \eqref{eq:iLK-a}. Substituting $(\phi, \delta \phi) = (\gamma(s), d\gamma(s)/ds)$ into~\eqref{cond:dLK-a} and integrating it over $s$ yields 
    \begin{align}\label{pf1:iLK}
        &k_1 \int_0^1 \left| \frac{d\gamma(s)}{ds}(0) \right|^p ds \nonumber\\
        &\le \int_0^1 V_D \left( t, \gamma(s), \frac{d\gamma(s)}{ds} \right) ds \le  k_2 \int_0^1 \left\| \frac{d\gamma(s)}{ds} \right\|_C^p ds.
    \end{align}
    Since taking the infimum preserves the inequalities, it follows from~\eqref{eq:dLK2iLK} that
    \begin{align*}
        &k_1 \inf_{\gamma \in \Gamma(\phi,\phi')} \int_0^1 \left| \frac{d\gamma(s)}{ds}(0) \right|^p ds\\
        &\le V_I(t,\phi,\phi') 
        \le  k_2 \inf_{\gamma \in \Gamma(\phi,\phi')} \int_0^1 \left\| \frac{d\gamma(s)}{ds} \right\|_C^p ds.
    \end{align*}
    The most right-hand side is upper bounded as
    \begin{align*}    
        &\inf_{\gamma \in \Gamma(\phi,\phi')} \int_0^1 \left\| \frac{d \gamma(s)}{ds} \right\|_C^p ds\\
        &\le \int_0^1 \left\| \frac{d (s \phi + (1-s)\phi') }{ds} \right\|_C^p ds
        = \|\phi - \phi' \|_C^p.
    \end{align*}
    {\color{black}From this and~\eqref{eq:dist}}, we obtain~\eqref{eq:iLK-a}.

    Next, we show that \eqref{eq:dLK2iLK} satisfies~\eqref{eq:iLK-b} also. For $t_0 \in \bR$ and $\gamma \in \Gamma(\phi,\phi')$ with $(\phi, \phi') \in C([-r,0],\Omega) \times C([-r,0],\Omega)$, consider the solution $\chi (t_0,\gamma (s))(t)$, $s \in [0,1]$ to a system~\eqref{eq:TDS} on $[t_0 -r, \infty)$, where recall that $\gamma (s) \in C([-r,0], \Omega)$ for each $s \in [0,1]$. Lemma~\ref{lem:sol} implies that
    \begin{align}\label{pf2:iLK}
    \delta x(t) = (D_\gamma \chi (t_0,\gamma (s))(t)) \delta \phi
    \end{align}
    is a solution to~\eqref{eq:vTDS} on $[t_0-r, \infty)$ for any $t_0 \in \bR$, $\gamma \in \Gamma(\phi,\phi')$, $(\phi, \phi') \in C([-r,0],\Omega) \times C([-r,0],\Omega)$, $s \in [0,1]$, and $\delta \phi \in C([-r,0],\bR^n)$. By selecting $\delta \phi = d \gamma (s)/ds$ and noting the chain rule 
    \begin{align}\label{pf2-5:iLK}
    D_s \chi (t_0,\gamma (s))(t) = (D_\gamma \chi (t_0,\gamma (s))(t)) \frac{d\gamma (s)}{ds},
    \end{align} 
    we conclude that
    \begin{align}\label{pf3:iLK}
    \delta x(t) = D_s \chi (t_0,\gamma (s))(t)
    \end{align}
    is a solution to~\eqref{eq:vTDS} on $[t_0-r, \infty)$ for any $t_0 \in \bR$, $\gamma \in \Gamma(\phi,\phi')$, $(\phi, \phi') \in C([-r,0],\Omega) \times C([-r,0],\Omega)$, and $s \in [0,1]$.
    
    Now, let $t \ge t_0$, and note that $\chi_t(t_0,\gamma(s)) \in \Gamma(\chi_t(t_0,\phi),\chi_t(t_0,\phi'))$ for any $\gamma \in \Gamma(\phi,\phi')$. Using the definition~\eqref{eq:dLK2iLK} of $V_I(t, \phi, \phi')$, we obtain the following
    \begin{align}\label{pf4:iLK}
        &V_I (t,\chi_t(t_0,\phi),\chi_t(t_0,\phi')) \nonumber\\
        &= \inf_{\bar \gamma \in \Gamma(\chi_t(t_0,\phi),\chi_t(t_0,\phi'))} \int_0^1 V_D \left(t, \bar \gamma(s), \frac{d\bar \gamma(s)}{ds}\right) ds \nonumber\\
        &\le \int_0^1 V_D (t, \chi_t(t_0,\gamma(s)), (D_s \chi (t_0,\gamma (s)))_t) ds
    \end{align}
    for any $t, t_0  \in \bR$ $(t \ge t_0)$, $\gamma \in \Gamma(\phi,\phi')$, and $(\phi, \phi') \in C([-r,0],\Omega) \times C([-r,0],\Omega)$.

    It follows from~\eqref{cond:dLK-b} and \eqref{pf3:iLK} that
    \begin{align*}
        &V_D(t,\chi_t(t_0,\gamma(s)), (D_s \chi (t_0,\gamma (s)))_t) \\
        &\qquad + k_3 \int_{t_0}^t |D_s \chi (t_0,\gamma (s))(\tau )|^p d\tau\\
        &\le V_D \left(t_0,\gamma(s), \frac{d\gamma (s)}{ds}\right).
    \end{align*}
    Combining this with \eqref{pf4:iLK} yields
    \begin{align}\label{pf5:iLK}
        &V_I (t,\chi_t(t_0,\phi),\chi_t(t_0,\phi')) \nonumber\\
        &\qquad + k_3 \int_0^1 \int_{t_0}^t |D_s \chi (t_0,\gamma (s))(\tau )|^p d\tau ds \nonumber\\
        &\le \int_0^1 V_D \left(t_0,\gamma(s), \frac{d\gamma (s)}{ds}\right) ds.
    \end{align}
    From~\eqref{pf1:iLK}, for each $t_0 \in \bR$, every $\gamma \in \Gamma(\phi,\phi')$, and $(\phi, \phi') \in C([-r,0],\Omega) \times C([-r,0],\Omega)$, the right-hand side is finite. From non-negativity of $V_I (t,\chi_t(t_0,\phi),\chi_t(t_0,\phi'))$, the second term of the left-hand side is finite for each $t, t_0 \in \bR$ $(t \ge t_0)$ and every $(\phi, \phi') \in C([-r,0],\Omega) \times C([-r,0],\Omega)$. Since any Riemann integral function on a bounded domain is Lebesgue integrable, and its Riemann and Lebesgue integrals are the same, e.g., \cite[Theorem 5.52]{WZ:15}, we can apply Fubini's theorem, e.g., \cite[Theorem 6.1]{WZ:15} to exchange the order of the integrations. Namely, we have
  \begin{align}\label{pf6:iLK}
        &\int_0^1 \int_{t_0}^t |D_s \chi (t_0,\gamma (s))(\tau )|^p d\tau ds \nonumber\\
        &= \int_{t_0}^t \int_0^1 |D_s \chi (t_0,\gamma (s))(\tau )|^p ds  d\tau
    \end{align}
    for each $t, t_0 \in \bR$ $(t \ge t_0)$,  every $\gamma \in \Gamma(\phi,\phi')$, and $(\phi, \phi') \in C([-r,0],\Omega) \times C([-r,0],\Omega)$. Next, applying~\eqref{pf1:R3} to~$D_s \chi (t_0,\gamma (s))(\tau )$ leads to
    \begin{align}\label{pf7:iLK}
    | \chi (t_0,\phi) (\tau ) - \chi (t_0,\phi') (\tau ) |^p
    \le \int_0^1 |D_s \chi (t_0,\gamma (s))(\tau )|^p ds.
    \end{align}
    This, \eqref{pf5:iLK}, and \eqref{pf6:iLK} yield
      \begin{align*}
        &V_I (t,\chi_t(t_0,\phi),\chi_t(t_0,\phi')) \\
        &\qquad + k_3 \int_{t_0}^t | \chi (t_0,\phi) (\tau ) - \chi (t_0,\phi') (\tau ) |^p  d\tau \\
        &\le \int_0^1 V_D \left(t_0,\gamma(s), \frac{d\gamma (s)}{ds}\right) ds.
    \end{align*}
    Since this holds for all $t, t_0 \in \bR$ $(t \ge t_0)$, $\gamma \in \Gamma(\phi,\phi')$, and $(\phi, \phi') \in C([-r,0],\Omega) \times C([-r,0],\Omega)$, taking the infimum with respect to $\gamma$ and using~\eqref{eq:dLK2iLK} lead to~\eqref{eq:iLK-b}.


    (Incremental $\Rightarrow$ Differential)
    Let $\Gamma'(\phi,\delta \phi)$ be the collection of the class $C^1$ paths $\gamma:[0,1] \to C([-r,0],\Omega)$ such that $\gamma(0)=\phi$ and $d\gamma(0)/ds=\delta \phi \in C([-r,0],\bR^n)$, where $d\gamma(0)/ds$ means the right derivative at $0$. We aim at showing that
    \begin{align}\label{eq:inc2fin}
        V_D(t,\phi,\delta \phi) := \inf_{\gamma \in \Gamma'(\phi,\delta \phi)} \limsup_{s \to 0+} \frac{V_I(t,\gamma(0),\gamma(s))}{s^p}
    \end{align}
    is a differential Lyapunov--Krasovskii functional.

    First, we show that \eqref{eq:inc2fin} satisfies~\eqref{cond:dLK-a}. Since taking the limit superior preserves the inequalities, it follows from~\eqref{eq:iLK-a} with $(\phi,\phi')=(\gamma(0), \gamma(s))$ that
    \begin{align*}
        &k_1 \limsup_{s \to 0+} \frac{|(\gamma(s))(0) - (\gamma(0))(0)|^p}{s^p}\\
        &\le \limsup_{s \to 0+} \frac{V_I(t,\gamma(0),\gamma(s))}{s^p} \le k_2 \limsup_{s \to 0+} \frac{\|\gamma(s) - \gamma(0)\|_C^p}{s^p}.
    \end{align*}
    Noting that taking the infimum also preserves the inequalities yields
    \begin{align*}
        &k_1 \inf_{\gamma \in \Gamma'(\phi,\delta \phi)} \limsup_{s \to 0+} \frac{|(\gamma(s))(0) - (\gamma(0))(0)|^p}{s^p}\\
        &\le V_D(t,\phi,\delta \phi) \le k_2 \inf_{\gamma \in \Gamma'(\phi,\delta \phi)} \limsup_{s \to 0+} \frac{\|\gamma(s) - \gamma(0)\|_C^p}{s^p},
    \end{align*}
    where~\eqref{eq:inc2fin} is used. For any $\gamma \in \Gamma'(\phi,\delta \phi)$, we have
    \begin{align*}
        &\lim_{s \to 0+} \frac{|(\gamma(s))(0) - (\gamma(0))(0)|^p}{s^p}\\
        &= \left( \left| \lim_{s \to 0+} \frac{(\gamma(s))(0) - (\gamma(0))(0)}{s}\right| \right)^p = |\delta \phi (0)|^p
    \end{align*}
    and
    \begin{align*}
        \lim_{s \to 0+} \frac{\|\gamma(s) - \gamma(0)\|_C^p}{s^p} 
        = \left(\left\| \lim_{s \to 0+} \frac{\gamma(s) - \gamma(0)}{s}\right\|_C \right)^p = \|\delta \phi\|_C^p.
    \end{align*}
    Hence,~\eqref{cond:dLK-a} holds for $V_D(t, \phi, \delta \phi)$ in~\eqref{eq:inc2fin}.

    Next, we consider showing~\eqref{cond:dLK-b}. From~\eqref{pf2:iLK}, $\delta x(t) = (D_\gamma \chi (t_0,\gamma (s))(t)) \delta \phi$ is a solution to~\eqref{eq:vTDS} on $[t_0-r, \infty)$ for any $t_0 \in \bR$, $\gamma \in \Gamma'(\phi,\delta \phi)$, $(\phi, \delta \phi) \in C([-r,0],\Omega) \times C([-r,0],\bR^n)$, and $s \in [0,1]$. In particular for $s \to 0^+$, this means
    \begin{align}
        (\chi_t(t_0,\phi), (D_\phi \chi(t_0,\phi))_t \delta \phi) \in \Gamma'(x_t,\delta x_t).
    \end{align}
    As a result, \eqref{eq:inc2fin} yields
    \begin{align*}
        &V_D(t,\chi_t(t_0,\phi), (D_\phi \chi(t_0,\phi))_t \delta \phi)\\
        &= \inf_{\bar \gamma \in \Gamma'(\chi_t(t_0,\phi), (D_\phi \chi(t_0,\phi))_t \delta \phi)} \limsup_{s \to 0^+} \frac{V_I(t,\bar \gamma(0),\bar \gamma(s))}{s^p}\\
        &\le \limsup_{s \to 0^+} \frac{V_I(t,\chi_t(t_0,\gamma(0)),\chi_t(t_0,\gamma(s)))}{s^p}
    \end{align*}
    for any $t, t_0 \in \bR$ $(t \ge t_0)$, $\gamma \in \Gamma'(\phi,\delta \phi)$, and $(\phi, \delta \phi) \in C([-r,0],\Omega) \times C([-r,0],\bR^n)$. It follows from~\eqref{eq:iLK-b} that
    \begin{align*}
        &V_D(t,\chi_t(t_0,\phi), (D_\phi \chi(t_0,\phi))_t \delta \phi)\\
        &\le \limsup_{s \to 0^+} \left(\frac{V_I(t_0,\gamma(0),\gamma(s))}{s^p}\right.\\
        &\qquad \qquad \left. - k_3 \int_{t_0}^t \frac{|\chi (t_0,\gamma(s))(\tau) - \chi (t_0,\gamma(0))(\tau)|^p}{s^p} d\tau\right)\\
        &\le \limsup_{s \to 0^+} \frac{V_I(t_0,\gamma(0),\gamma(s))}{s^p}\\
        &\qquad - k_3 \liminf_{s \to 0^+} \int_{t_0}^t \frac{|\chi (t_0,\gamma(s))(\tau) - \chi (t_0,\gamma(0))(\tau)|^p}{s^p} d\tau.
    \end{align*}
    We consider the corresponding Lebesgue integral to the second term of the most right-hand side. By Fatou lemma~\cite[Theorem 5.34]{WZ:15}, we have
    \begin{align*}
        &\liminf_{s \to 0^+} \int_{t_0}^t \frac{|\chi (t_0,\gamma(s))(\tau) - \chi (t_0,\gamma(0))(\tau)|^p}{s^p} d\tau\\
        &\ge \int_{t_0}^t \liminf_{s \to 0^+} \frac{|\chi (t_0,\gamma(s))(\tau) - \chi (t_0,\gamma(0))(\tau)|^p}{s^p} d\tau\\
        &= \int_{t_0}^t \left| \lim_{s \to 0^+} \frac{\chi (t_0,\gamma(s))(\tau) - \chi (t_0,\gamma(0))(\tau)}{s}\right|^p d\tau\\
        &= \int_{t_0}^t | D_s \chi (t_0,\gamma(0))(\tau) |^p d\tau = \int_{t_0}^t | D_\phi \chi (t_0,\phi)(\tau) \delta \phi |^p d\tau,
    \end{align*}
    where~\eqref{pf2-5:iLK} is used in the last equality.
    This is equivalent to the Riemann integral because this is non-negative and upper bounded. Therefore, we obtain
        \begin{align*}
        &V_D(t,\chi_t(t_0,\phi), (D_\phi \chi(t_0,\phi))_t \delta \phi) \\
        & + k_3 \int_{t_0}^t | D_\phi \chi (t_0,\phi)(\tau) \delta \phi |^p d\tau
        \le \limsup_{s \to 0^+} \frac{V_I(t_0,\gamma(0),\gamma(s))}{s^p}.
    \end{align*}
    Since this holds for all $t, t_0 \in \bR$ $(t \ge t_0)$, $\gamma \in \Gamma'(\phi, \delta \phi)$, and $(\phi, \delta \phi) \in C([-r,0],\Omega) \times C([-r,0],\bR^n)$, taking the infimum with respect to $\gamma$ and using~\eqref{eq:inc2fin} lead to~\eqref{cond:dLK-b}.
\QED


\subsection{Proof of Theorem~\ref{thm:dLR}}\label{app:dLR}
    {\color{black}Applying~\cite[Theorem 2]{EA:20}, one can show that there exist~$0 < \lambda \in \bR$ and~$1 \le k \in \bR$ satisfying
    \begin{align}\label{pf1:dLR_ex}
     |(D_\phi \chi (t_0,\phi)(t)) \delta \phi| 
    \le k e^{-\lambda (t-t_0)} \|\delta \phi\|_C, 
    \; \forall t \ge t_0
    \end{align}
    for any $t_0 \in \bR$ and $(\phi, \delta \phi) \in C([-r,0],\Omega) \times C([-r,0],\bR^n)$.}
    For each $(\phi, \phi') \in C([-r,0],\Omega) \times C([-r,0],\Omega)$, consider the path $\gamma (s) = s\phi + (1-s)\phi'$, $s \in [0, 1]$. As shown in the proof of Theorem~\ref{thm:iLK}, $\delta x(t) = D_s \chi (t_0,\gamma (s))(t)$ is a solution to the variational system~\eqref{eq:vTDS} on $[t_0-r, \infty)$ for any $t_0 \in \bR$, $(\phi, \phi') \in C([-r,0],\Omega) \times C([-r,0],\Omega)$, and $s \in [0,1]$. Thus, substituting $\phi = \gamma (s)$, $s \in [0,1]$, and $\delta \phi = d\gamma (s)/ds = \phi - \phi'$ into \eqref{pf1:dLR_ex} yields
    {\color{black}
    \begin{align*}
    |(D_s \chi (t_0,\gamma (s)))(t)| 
    \le k e^{-\lambda (t-t_0)} \| \phi - \phi' \|_C, 
    \; \forall t \ge t_0
    \end{align*}
     for any $t_0 \in \bR$, $(\phi, \phi') \in C([-r,0],\Omega) \times C([-r,0],\Omega)$, and $s \in [0, 1]$. Integrating both sides with respect to~$s$ over~$[0, 1]$ and utilizing~\eqref{eq:dist} lead to uniform incremental exponential stability.}
\QED





\section*{References}\vspace*{-4mm}
\bibliographystyle{IEEEtran}
\bibliography{ref}

@article{Han:91,
  title = {Retarded {{Dynamical Systems}}: {{Stability}} and {{Characteristic Functions}} ({{G}}. {{St\'ep\'an}})},
  shorttitle = {Retarded {{Dynamical Systems}}},
  author = {Hannsgen, Kenneth B.},
  year = 1991,
  journal = {SIAM Review},
  volume = {33},
  number = {1},
  pages = {147--147},
  publisher = {{Society for Industrial and Applied Mathematics}}
}

@article{LS:98,
	Author = {W. Lohmiller and J.-J. E. Slotine},
	Journal = {Automatica},
	Number = {6},
	Pages = {683--696},
	Title = {On contraction analysis for non-linear systems},
	Volume = {34},
	Year = {1998}
}

@article{Angeli:02,
	Author = {D. Angeli},
	Journal = {IEEE Trans. Automat. Control},
	Number = {3},
	Pages = {410--421},
	Publisher = {IEEE},
	Title = {A {L}yapunov approach to incremental stability properties},
	Volume = {47},
	Year = {2002}
}

@article{Mon:03,
  title = {Oscillatory {{Expression}} of {{Hes1}}, P53, and {{NF-$\kappa$B Driven}} by {{Transcriptional Time Delays}}},
  author = {Monk, Nicholas A. M.},
  year = 2003,
  journal = {Current Biology},
  volume = {13},
  number = {16},
  pages = {1409--1413}
}

@article{WS:06,
  title={Contraction analysis of time-delayed communications and group cooperation},
  author={Wang, Wei and Slotine, J.-J. E.},
  journal={IEEE Trans. Automat. Control},
  volume={51},
  number={4},
  pages={712--717},
  year={2006},
  publisher={IEEE}
}

@article{Karafyllis:06,
  title = {Lyapunov Theorems for Systems Described by Retarded Functional Differential Equations},
  author = {Karafyllis, Iasson},
  year = 2006,
  journal = {Nonlinear Anal. Theory Methods Appl.},
  volume = {64},
  number = {3},
  pages = {590--617}
}

@article{Pepe:07,
  title={The problem of the absolute continuity for {L}yapunov--{K}rasovskii functionals},
  author={Pepe, Pierdomenico},
  journal={IEEE Trans. Automat. Control},
  volume={52},
  number={5},
  pages={953--957},
  year={2007},
  publisher={IEEE}
}

@article{IPZ:08,
  title={Global output stability for systems described by retarded functional differential equations: {L}yapunov characterizations},
  author={Karafyllis, Iasson and Pepe, Pierdomenico and Jiang, Zhong-Ping},
  journal={Eur. J. Control},
  volume={14},
  number={6},
  pages={516--536},
  year={2008},
  publisher={Elsevier}
}

@book{Gan:09,
  title = {Economic {{Dynamics}}},
  author = {Gandolfo, Giancarlo},
  year = 2009,
  edition = {4th ed.},
  publisher = {Springer}
}

@article{PPD:10,
  title={Symbolic models for nonlinear time-delay systems using approximate bisimulations},
  author={Pola, Giordano and Pepe, Pierdomenico and Di Benedetto, Maria D and Tabuada, Paulo},
  journal={Syst. Control Lett.},
  volume={59},
  number={6},
  pages={365--373},
  year={2010},
  publisher={Elsevier}
}

@article{CPR:13,
  title={A {R}azumikhin approach for the incremental stability of delayed nonlinear systems},
  author={A. Chaillet and A. Y. Pogromsky and B. R{\"u}ffer},
  journal={Proc. 52nd IEEE Conf. Decis. Control},
  pages={1596--1601},
  year={2013},
  organization={IEEE}
}

@book{HL:13,
  title = {Introduction to Functional Differential Equations},
  author = {J. K. Hale and S. M. V. Lunel},
  year = {1993},
  publisher = {Springer Science+Businesss Media},
  address = {New York}
}

@article{PK:13,
  title={Converse {L}yapunov--{K}rasovskii theorems for systems described by neutral functional differential equations in {H}ale's form},
  author={Pepe, Pierdomenico and Karafyllis, Iasson},
  journal={Int. J. Control},
  volume={86},
  number={2},
  pages={232--243},
  year={2013},
  publisher={Taylor \& Francis}
}

@article{FS:14,
  author={F. Forni and R. Sepulchre},
  journal={IEEE Trans. Automat. Control}, 
  title={A Differential {L}yapunov Framework for Contraction Analysis}, 
  year={2014},
  volume={59},
  number={3},
  pages={614-628}
}

@book{Fridman:14b,
  title = {Introduction to {{Time-Delay Systems}}: {{Analysis}} and {{Control}}},
  shorttitle = {Introduction to {{Time-Delay Systems}}},
  author = {Fridman, Emilia},
  year = {2014},
  publisher = {Birkhaeuser}
}

@book{WZ:15,
  title={Measure and Integral: An Introduction to Real Analysis},
  author={R. L. Wheeden and A. Zygmund},
  year={2015},
  edition={2},
  publisher={CRC Press},
  address = {Boca Raton}
}

@ARTICLE{MS:17,
  author={Manchester, Ian R. and Slotine, Jean-Jacques E.},
  journal={IEEE Trans. Automat. Control}, 
  title={Control Contraction Metrics: Convex and Intrinsic Criteria for Nonlinear Feedback Design}, 
  year={2017},
  volume={62},
  number={6},
  pages={3046-3053},
  }

@article{PKJ:17,
  title={Lyapunov--Krasovskii characterization of the input-to-state stability for neutral systems in {H}ale’s form},
  author={Pepe, Pierdomenico and Karafyllis, Iasson and Jiang, Zhong-Ping},
  journal={Syst. Control Lett.},
  volume={102},
  pages={48--56},
  year={2017},
  publisher={Elsevier}
}

@article{NT:18,
  title={On contraction of functional differential equations},
  author={Ngoc, Pham Huu Anh and Trinh, Hieu},
  journal={SIAM J. Control Optim.},
  volume={56},
  number={3},
  pages={2377--2397},
  year={2018},
  publisher={SIAM}
}

@article{LSF:19,
  title = {Survey on Time-Delay Approach to Networked Control},
  author = {Liu, Kun and Selivanov, Anton and Fridman, Emilia},
  year = 2019,
  journal = {Annu. Rev. Control},
  volume = {48},
  pages = {57--79}
}

@article{EA:20,
  title={On estimation of rates of convergence in {L}yapunov--{R}azumikhin approach},
  author={Efimov, Denis and Aleksandrov, Alexander},
  journal={Automatica},
  volume={116},
  pages={108928},
  year={2020},
  publisher={Elsevier}
}

@article{LR:20,
  title = {Stability Analysis by Contraction Principle for Impulsive Systems with Infinite Delays},
  author = {Liu, Xinzhi and Ramirez, Cesar},
  year = 2020,
  journal = {Commun. Nonlinear Sci. Numer. Simul.},
  volume = {82},
  pages = {105021}
}

@book{Bullo:22,
  author    = {F. Bullo},
  publisher = {Kindle Direct Publishing},
  title     = {Contraction Theory for Dynamical Systems},
  year      = {2022},
  edition   = {1.0},
}

@ARTICLE{GAA:22,
  author={Giaccagli, Mattia and Astolfi, Daniele and Andrieu, Vincent and Marconi, Lorenzo},
  journal={IEEE Trans. Automat. Control}, 
  title={Sufficient Conditions for Global Integral Action via Incremental Forwarding for Input-Affine Nonlinear Systems}, 
  year={2022},
  volume={67},
  number={12},
  pages={6537-6551},
}

@article{KB:24,
  title={Incremental versus differential approaches to exponential stability and passivity},
  author={Kawano, Yu and Besselink, Bart},
  journal={IEEE Trans. Automat. Control},
  year={2024},
  volume={69},
  number={9},
  pages={6450-6457}
}

@article{DCG:25,
  title={Time-varying convex optimization: A contraction and equilibrium tracking approach},
  author={Davydov, Alexander and Centorrino, Veronica and Gokhale, Anand and Russo, Giovanni and Bullo, Francesco},
  journal={IEEE Trans. Automat. Control},
  year={2025},
  publisher={IEEE},
  doi={10.1109/TAC.2025.3576043}}

@ARTICLE{KH:25,
  author={Y. Kawano and Y. Hosoe},
  journal={IEEE Control Syst. Lett.}, 
  title={Contraction Analysis of Almost Surely Monotone Discrete-Time Systems With Unknown Time Delay}, 
  year={2025},
  volume={9},
  number={},
  pages={426-431},
}

@article{KS:25,
  title={Virtual Contraction Approach to Decentralized Adaptive Stabilization of Nonlinear Time-Delayed Networks},
  author={Kawano, Yu and Sun, Zhiyong},
  journal={IEEE Trans. Automat. Control},
  year={2026},
  note = {(early access)}
}

@ARTICLE{KSS:25,
  author={Kawano, Yu and van der Schaft, A. J. and Scherpen, J. M. A.},
  journal={IEEE Trans. Automat. Control}, 
  title={Youla–{K}u\~cera Parameterization in Contraction Framework}, 
  year={2025},
  volume={70},
  number={3},
  pages={1667-1682},
  }


\end{document}